\documentclass[a4paper,11pt]{article}

\usepackage{fullpage}
\usepackage{times}
\usepackage{soul}
\usepackage{url}
\usepackage[utf8]{inputenc}
\usepackage[small]{caption}
\usepackage{graphicx}
\usepackage{xcolor}
\usepackage{amsmath}
\usepackage{booktabs}
\usepackage{algorithm}
\usepackage{enumerate}
\usepackage{dsfont}
\usepackage{cleveref}

\usepackage{amsthm}
\usepackage{amssymb}
\usepackage{algpseudocode}

\usepackage{bbm}

\newtheorem{theorem}{Theorem}[section]
\newtheorem{corollary}[theorem]{Corollary}
\newtheorem{proposition}[theorem]{Proposition}
\newtheorem{lemma}[theorem]{Lemma}

\theoremstyle{definition}
\newtheorem{definition}[theorem]{Definition}
\newtheorem{example}[theorem]{Example}

\DeclareMathOperator*{\argmax}{arg\,max}
\DeclareMathOperator*{\argmin}{arg\,min}
\newcommand*\diff{\mathop{}\!\mathrm{d}} 

\newcommand{\mgain}[3]{\Delta^+_{#1}(#2, #3)}
\newcommand{\mgaini}[2]{\mgain{i}{#1}{#2}}
\newcommand{\mloss}[3]{\Delta^-_{#1}(#2, #3)}
\newcommand{\mlossi}[2]{\mloss{i}{#1}{#2}}

\newcommand{\mwhw}{MWHW$_x$}

\usepackage{natbib}
\usepackage{authblk}

\allowdisplaybreaks

\title{\bf Weighted Envy-Freeness for Submodular Valuations}

\author[1]{Luisa Montanari}
\author[1]{Ulrike Schmidt-Kraepelin}
\author[2]{Warut Suksompong}
\author[3]{Nicholas Teh}

\affil[1]{Technische Universität Berlin, Germany}
\affil[2]{National University of Singapore, Singapore}
\affil[3]{University of Oxford, UK}

\usepackage[colorinlistoftodos,prependcaption,textsize=tiny]{todonotes}

\date{\vspace{-10mm}}

\begin{document}

\maketitle

\begin{abstract}
We investigate the fair allocation of indivisible goods to agents with possibly different entitlements represented by weights. 
Previous work has shown that guarantees for additive valuations with existing envy-based notions cannot be extended to the case where agents have matroid-rank (i.e., binary submodular) valuations.
We propose two families of envy-based notions for matroid-rank and general submodular valuations, one based on the idea of transferability and the other on marginal values.
We show that our notions can be satisfied via generalizations of rules such as picking sequences and maximum weighted Nash welfare.
In addition, we introduce welfare measures based on harmonic numbers, and show that variants of maximum weighted harmonic welfare offer stronger fairness guarantees than maximum weighted Nash welfare under matroid-rank valuations.
\end{abstract}

\section{Introduction}

Fair division refers to the study of how to fairly allocate resources among agents with possibly differing preferences.
Over the 70 years since \citet{Steinhaus48} initiated a mathematical framework of fair division, the field has given rise to numerous fairness notions and procedures for computing fair outcomes in a variety of scenarios \citep{BramsTa96,Moulin03}.
For instance, in the common scenario of allocating indivisible goods, the notion \emph{envy-freeness up to one good (EF1)} has emerged as a standard benchmark.
An allocation of the goods satisfies EF1 if any envy that an agent has toward another agent can be eliminated by removing some good in the latter agent's bundle.
Even when agents have arbitrary monotonic valuations over the goods, an EF1 allocation always exists and can be found in polynomial time \citep{LiptonMaMo04}.

The definitions of many fairness notions in the literature, including EF1, inherently assume that all agents have the same entitlement to the resource.
In the past few years, several researchers have examined a more general model in which different agents may have different \emph{weights} reflecting their entitlements to the goods \citep{FarhadiGhHa19,AzizMoSa20,BabaioffEzFe21,BabaioffNiTa21,ChakrabortyIgSu21,ChakrabortyScSu21,ChakrabortySeSu22,GargKuKu20,GargHuVe21,GargHuMu22,ScarlettTeZi21,HoeferScVa22,SuksompongTe22,ViswanathanZi22}.
This model allows us to capture settings such as inheritance division, in which relatives are typically entitled to unequal shares of the legacy, as well as resource allocation among groups or organizations of different sizes.
\citet{ChakrabortyIgSu21} generalized EF1 to \emph{weighted EF1 (WEF1)}: for instance, if Alice's weight is three times as high as Bob's, then WEF1 stipulates that, after removing some good from Bob's bundle, Alice should have at least three times as much value for her own bundle as for Bob's.
The same authors demonstrated that if agents have \emph{additive} valuations over the goods, a complete WEF1 allocation always exists and can be computed efficiently.\footnote{An allocation is called \emph{complete} if it allocates all of the goods.}
However, they provided the following example showing that existence is no longer guaranteed once we move beyond additivity.
\begin{example}[\citet{ChakrabortyIgSu21}]
\label{ex:WEF1}
Consider an instance with $n = 2$ agents whose weights are $w_1 = 1$ and $w_2 = 2$, and $m \ge 6$ goods. Agent~$1$ has an additive valuation with value~$1$ for every good, whereas agent~$2$ has value $0$ for the empty bundle and $1$ for any nonempty bundle.
If agent~$1$ is allocated more than one good, then agent~$2$ has weighted envy toward agent~$1$ even after removing any good from agent~$1$'s bundle. 
Thus, assuming that all goods need to be allocated, agent~$2$ obtains at least $m-1$ goods. 
Again, this causes weighted envy according to WEF1, this time from agent~$1$ toward agent~$2$. 
Hence, no complete WEF1 allocation exists in this instance.
\end{example}
The impossibility result illustrated in this example still holds even if WEF1 is relaxed to \emph{weak WEF1 (WWEF1)}, whereby an agent is allowed to either remove a good from the other agent's bundle or copy one such good into her own bundle,\footnote{In fact, the impossibility persists even with WWEF$c$ for any constant $c$; see the discussion after Proposition~8.1 of \citet{ChakrabortyIgSu21}.} and stands in contrast to the aforementioned EF1 guarantee in the unweighted setting \citep{LiptonMaMo04}.
In light of these observations, \citet{ChakrabortyIgSu21} left open the direction of identifying appropriate envy-based notions for non-additive valuations.
Note also that the valuations in Example~\ref{ex:WEF1} are particularly simple: both agents have \emph{binary submodular} valuations, that is, submodular valuations\footnote{See the definition of submodularity in \Cref{sec:prelims}.} in which the marginal gain from receiving any single good is either $0$ or $1$.
Binary submodular valuations are also known as \emph{matroid-rank valuations}, and have been studied in a number of recent fair division papers, mostly in the unweighted setting \citep{BabaioffEzFe21-dichotomous,BarmanVe21,BarmanVe22,BenabbouChIg21,GokoIgKa22,ViswanathanZi22,ViswanathanZi22-unweighted}.\footnote{An exception is the recent work of \citet{ViswanathanZi22}, which deals with the weighted setting.}
Such valuations arise in settings such as the allocation of course slots to students, or apartments in public housing estates to ethnic groups \citep{BenabbouChIg21}.

In this paper, we explore weighted envy-freeness for both matroid-rank and general submodular valuations.
We propose new envy-based notions and show that they can be satisfied in these settings, not only via extensions of existing algorithms, but also via new rules.
For the sake of generality, we define our notions based on the notion WEF$(x,1-x)$ of \citet{ChakrabortySeSu22}.
With additive valuations, given a parameter $x\in [0,1]$, WEF$(x,1-x)$ allows each agent~$i$ to subtract $x$ times the value of some good in another agent~$j$'s bundle from $i$'s value for this bundle, and add $(1-x)$ times the value of this good to the value of $i$'s own bundle.
WEF1 corresponds to WEF$(1,0)$, and higher values of $x$ yield notions that favor lower-weight agents.
\citet{ChakrabortySeSu22} showed that for any instance with additive valuations and any $x$, a complete WEF$(x,1-x)$ allocation always exists; on the other hand, they proved that for any distinct $x$ and $x'$, there is an instance with binary additive valuations and identical goods in which no complete allocation satisfies both WEF$(x,1-x)$ and WEF$(x',1-x')$.

\subsection{Our Contributions}

In \Cref{sec:prelims}, we introduce two new families of weighted envy-freeness notions.
The first family, \emph{TWEF$(x,1-x)$}, is based on the concept of ``transferability'':\footnote{This concept has been discussed by \citet{BenabbouChIg21} and \citet{ChakrabortyIgSu21}.} we consider the condition TWEF$(x,1-x)$ from agent~$i$ to agent~$j$ to be violated only if the WEF$(x,1-x)$ condition between $i$ and $j$ fails \emph{and} $i$'s value for her own bundle does not increase even if all goods from $j$'s bundle are transferred to $i$'s bundle.
TWEF$(x,1-x)$ handles instances such as the one in Example~\ref{ex:WEF1}, where an agent could be unsatisfied with respect to WEF$(x,1-x)$ even if she already receives her maximum possible utility.
Our second family, \emph{WMEF$(x,1-x)$}, is an extension of the notion \emph{marginal EF1 (MEF1)} of \citet{CaragiannisKuMo19} from the unweighted setting.
The idea is that, instead of agent~$i$ considering her value for agent~$j$'s bundle as in WEF$(x,1-x)$, agent~$i$ considers her \emph{marginal} value of $j$'s bundle when added to $i$'s own bundle.
While TWEF$(x,1-x)$ is stronger than WMEF$(x,1-x)$, we show that the former notion is suitable primarily for matroid-rank valuations, whereas the latter can be guaranteed even for general submodular valuations. 

In Sections~\ref{sec:pickseq} and \ref{sec:Nash}, we allow agents to have arbitrary submodular valuations.
In Section~\ref{sec:pickseq}, we investigate \emph{picking sequences}, which let agents take turns picking a good according to a specified agent ordering until the goods run out.
While previous work on picking sequences typically assumed that agents have additive valuations, this assumption may be violated in real-world applications of picking sequences, such as the allocation of ministries to political parties. 
We adjust picking sequences to submodular valuations by letting agents pick a good with the highest \emph{marginal} gain in each of their turns.
We show that for every $x$, the output of the adjusted version of the picking sequence proposed by \citet{ChakrabortySeSu22} with parameter~$x$ satisfies WMEF$(x,1-x)$; this generalizes their result from the weighted additive setting.
As a corollary, in the unweighted submodular setting, the adjusted version of the commonly studied round-robin algorithm produces an MEF1 allocation.
In Section~\ref{sec:Nash}, we  consider the \emph{maximum weighted Nash welfare (MWNW)} rule, which chooses an allocation that maximizes the weighted product of the agents' utilities.
Although prior results rule out the possibility that MWNW implies WMEF$(x,1-x)$ for any~$x$, we show that an MWNW allocation always satisfies a relaxation of WMEF$(x,1-x)$ called \emph{weak weighted MEF1 (WWMEF1)}.
This extends a corresponding result of \citet{ChakrabortyIgSu21} from the weighted additive setting, which is in turn a generalization of the prominent result by \citet{CaragiannisKuMo19} in the unweighted additive setting.

Next, in Sections~\ref{sec:transfer} and \ref{sec:harmonic}, we focus on agents with matroid-rank valuations.
In Section~\ref{sec:transfer}, we extend the ``transfer algorithm'' of \citet{BenabbouChIg21} from the unweighted setting, and prove that our algorithm returns a clean\footnote{An allocation is \emph{clean} if no good can be discarded from an agent's bundle without decreasing the agent's utility \citep{BenabbouChIg21}. The term \emph{non-redundant} has also been used with the same meaning \citep{BabaioffEzFe21-dichotomous}.} TWEF$(x,1-x)$ allocation that maximizes the unweighted utilitarian welfare.
While Benabbou et al.'s potential function argument can be generalized to show that our algorithm terminates, it is insufficient for establishing polynomial-time termination in our setting with different weights; hence, we devise a more elaborate argument for this purpose.
Finally, in Section~\ref{sec:harmonic}, we introduce new welfare measures based on harmonic numbers and their variants.\footnote{The harmonic welfare measure is the basis of the \emph{proportional approval voting (PAV)} rule in the context of \emph{approval-based committee voting} \citep{LacknerSk22}.
To the best of our knowledge, we are the first to consider this measure in the context of fair division.}
Perhaps surprisingly, we demonstrate that the maximum-welfare rules based on our measures offer stronger fairness guarantees compared to MWNW under matroid-rank valuations.
In particular, while MWNW does not imply WEF$(x,1-x)$ for any $x$ even with binary additive valuations and identical goods \citep{ChakrabortySeSu22}, we prove that a clean \emph{maximum weighted harmonic welfare} allocation parameterized by $x$ satisfies TWEF$(x,1-x)$ for matroid-rank valuations (and therefore WEF$(x,1-x)$ for binary additive valuations).\footnote{To further exhibit the potential of harmonic welfare, we show in Appendix~\ref{app:harmonic-more} that, in the unweighted additive setting, if each agent's value for each good is an integer, then a maximum harmonic welfare allocation always satisfies EF1.}

\section{Preliminaries}
\label{sec:prelims}

Let $N = [n]$ be the set of agents and $G = \{g_1,\dots,g_m\}$ be the set of indivisible goods, where $[k] := \{1,\dots,k\}$ for any positive integer $k$.
A \emph{bundle} refers to a subset of $G$.
Each agent $i\in N$ has a \emph{weight} $w_i > 0$ representing her entitlement, and a \emph{valuation function} (or \emph{utility function}) $v_i:2^G\rightarrow\mathbb{R}_{\ge 0}$.
The setting where all of the weights are equal is sometimes referred to as the \emph{unweighted setting}.
For convenience, we write $v_i(g)$ instead of $v_i(\{g\})$ for a single good~$g$.
We assume throughout the paper that $v_i$ is 
\begin{itemize}
\item \emph{monotone}: $v_i(G') \le v_i(G'')$ for all $G'\subseteq G''\subseteq G$;
\item \emph{submodular}: $v_i(G'\cup\{g\}) - v_i(G') \ge v_i(G''\cup \{g\}) - v_i(G'')$ for all $G'\subseteq G''\subseteq G$ and $g\in G\setminus G''$;
\item \emph{normalized}: $v_i(\emptyset) = 0$.
\end{itemize}
The function $v_i$ is said to be \emph{matroid-rank} (or \emph{binary submodular}) if it is submodular and $v_i(G'\cup\{g\}) - v_i(G')\in \{0,1\}$ for all $G'\subseteq G$ and $g\in G\setminus G'$.
Moreover, $v_i$ is \emph{additive} if $v_i(G') = \sum_{g\in G'}v_i(g)$ for all $G'\subseteq G$, and \emph{binary additive} if it is additive and $v_i(g) \in \{0,1\}$ for all $g\in G$.
An \emph{instance} consists of the set of agents~$N$, the set of goods~$G$, and the agents' weights $(w_i)_{i\in N}$ and valuation functions $(v_i)_{i\in N}$.

An \emph{allocation}~$\mathcal{A}$ is a list of bundles $(A_1,\dots,A_n)$ such that no two bundles overlap, where bundle~$A_i$ is assigned to agent~$i$.
The allocation is \emph{complete} if $\bigcup_{i\in N}A_i = G$.
It is \emph{Pareto-optimal (PO)} if there does not exist another allocation $\mathcal{A}'$ such that $v_i(A_i') \ge v_i(A_i)$ for all $i\in N$ and the inequality is strict for at least one $i\in N$; such an allocation~$\mathcal{A}'$ is said to \emph{Pareto-dominate} $\mathcal{A}$.
For each $i\in N$, we denote by $N^+_\mathcal{A}$ the subset of agents receiving positive utility from $\mathcal{A}$.
The \emph{unweighted utilitarian welfare} of $\mathcal{A}$ is defined as $\sum_{i\in N}v_i(A_i)$.

For a bundle $G'\subseteq G$, we define the \emph{marginal gain} of a good $g\not\in G'$ for agent~$i$ as $\mgaini{G'}{g} := v_i(G' \cup \{g\})-v_i(G')$. 
Similarly, the \emph{marginal loss} of a good $g \in G'$ for agent~$i$ is defined as $\mlossi{G'}{g} := v_i(G') - v_i(G'\setminus\{g\})$. 
An allocation $\mathcal{A}$ is called \emph{clean} (or \emph{non-redundant}) if for any $i\in N$ and any $g\in A_i$, it holds that $\mlossi{A_i}{g} > 0$.
For matroid-rank valuations, $\mathcal{A}$ is clean if and only if $v_i(A_i) = |A_i|$ for all $i\in N$ \citep[Prop.~3.3]{BenabbouChIg21}.

We now introduce our first family of fairness notions, TWEF$(x,y)$.\footnote{This is not to be confused with the notion that \citet{ChakrabortyIgSu21} informally termed ``transfer WEF1'', which allows an agent to transfer a good from another agent's bundle over to her own bundle (instead of only making a copy of it as in weak WEF1), or with the notion ``typewise (M)EF1'' of \citet{BenabbouChEl19} in a different setting.}

\begin{definition}[TWEF$(x,y)$]
For $x,y\in[0,1]$, an allocation $\mathcal{A}$ is said to satisfy \emph{transferable WEF$(x,y)$ (TWEF$(x,y)$)} if, for any pair of agents $i, j \in N$, either $v_i(A_i) = v_i(A_i \cup A_j)$ or there exists $g\in A_j$ such that 
\[
\frac{v_i(A_i) + y\cdot \mgaini{A_i}{g}}{w_i} \geq \frac{v_i(A_j) - x\cdot \mlossi{A_j}{g}}{w_j}.
\]
\end{definition}

By submodularity, the condition $v_i(A_i) = v_i(A_i\cup A_j)$ is equivalent to the condition that $v_i(A_i) = v_i(A_i\cup\{g\})$ for every $g\in A_j$.

For any $x$ and $y$, if valuations are additive, then TWEF$(x,y)$ reduces to the notion WEF$(x,y)$ of \citet{ChakrabortySeSu22}.
We will mostly be concerned with the case where $y = 1-x$.
As we will see, TWEF$(x,1-x)$ is a useful notion for matroid-rank valuations.
However, like WEF$(x,1-x)$, it can be too demanding for general submodular valuations.
For instance, in Example~\ref{ex:WEF1}, if agent~$2$ has value $1+(|G'|-1) \cdot \varepsilon$ for any nonempty bundle $|G'|$, where $\varepsilon > 0$ is a small constant, then the condition $v_i(A_i) = v_i(A_i\cup A_j)$ becomes impotent and a complete TWEF$(x,1-x)$ allocation does not exist for any~$x$.
The second family of notions that we propose, which is based on the \emph{marginal EF1 (MEF1)} notion of \citet{CaragiannisKuMo19},\footnote{In the unweighted setting, an allocation satisfies MEF1 if for all $i,j\in N$, there exists $g\in A_j$ such that $v_i(A_i) \ge v_i(A_i\cup A_j\setminus\{g\}) - v_i(A_i)$.} does not suffer from this shortcoming.

\begin{definition}[WMEF$(x,y)$]
For $x,y\in[0,1]$, an allocation $\mathcal{A}$ is said to satisfy \emph{WMEF$(x,y)$} if, for any pair of agents $i, j \in N$, either $A_j = \emptyset$ or there exists $g\in A_j$ such that 
\[
\frac{v_i(A_i) + y\cdot \mgaini{A_i}{g}}{w_i} \geq \frac{v_i(A_i\cup A_j) - v_i(A_i) - x\cdot \mlossi{A_i\cup A_j}{g}}{w_j}.
\]
\end{definition}

If valuations are additive, WMEF$(x,y)$ reduces to WEF$(x,y)$ for any $x$ and $y$.
On the other hand, if all agents have the same weight, WMEF$(x,1-x)$ reduces to MEF1 only if $x = 1$.
The following proposition establishes an implication relationship between our two families of notions.

\begin{proposition}
\label{prop:TWEF-WMEF}
For $x,y\in[0,1]$, any TWEF$(x, y)$ allocation is also WMEF$(x, y)$. 
\end{proposition}

\begin{proof}
Let $\mathcal{A}$ be a TWEF$(x,y)$ allocation, and consider any $i,j\in N$. 
By definition of TWEF$(x,y)$, either $v_i(A_i) = v_i(A_i\cup A_j)$ or there exists $g\in A_j$ such that $\frac{v_i(A_i) + y\cdot \mgaini{A_i}{g}}{w_i} \geq \frac{v_i(A_j) - x\cdot \mlossi{A_j}{g}}{w_j}$.
Assume first that the latter holds.
Since $v_i$ is submodular, we have $v_i(A_i) + v_i(A_j\setminus\{g\}) \ge v_i(A_i\cup A_j\setminus\{g\})$ and $\mlossi{A_j}{g} \geq \mlossi{A_i \cup A_j}{g}$. 
It follows that 
\begin{align*}
    \frac{v_i(A_i) + y\cdot\mgaini{A_i}{g}}{w_i} &\geq \frac{v_i(A_j)-x\cdot\mlossi{A_j}{g}}{w_j}\\
    &=\frac{v_i(A_j\setminus\{g\}) + (1-x)\cdot\mlossi{A_j}{g}}{w_j} \\
    &\geq \frac{v_i(A_i \cup A_j\setminus\{g\})-v_i(A_i)+(1-x)\cdot\mlossi{A_i \cup A_j}{g}}{w_j}\\ 
    &= \frac{v_i(A_i\cup A_j) - v_i(A_i) - x\cdot\mlossi{A_i \cup A_j}{g}}{w_j}.
\end{align*}
Hence, the WMEF$(x,y)$ condition between agents $i$ and $j$ is fulfilled.

Next, assume that $v_i(A_i) = v_i(A_i \cup A_j)$.
If $A_j = \emptyset$, then the WMEF$(x,y)$ condition between agents $i$ and $j$ is automatically fulfilled.
Otherwise, for any good $g\in A_j$, we have
\begin{align*}
\frac{v_i(A_i)+y\cdot\mgaini{A_i}{g}}{w_i} 
\geq 0 
&\geq \frac{-x\cdot\mlossi{A_i \cup A_j}{g}}{w_j} \\
&= \frac{v_i(A_i \cup A_j) - v_i(A_i)-x\cdot\mlossi{A_i \cup A_j}{g}}{w_j},
\end{align*}
and WMEF$(x,y)$ between $i$ and $j$ is again fulfilled.
\end{proof}

Since the valuations that we consider in this paper are not necessarily additive, in order to reason about the running time of algorithms, we make the standard assumption that an algorithm can query the value of any agent~$i$ for any bundle~$G'\subseteq G$ in constant time.

\section{Picking Sequences}
\label{sec:pickseq}

In this section, we investigate \emph{picking sequences}, which are procedures wherein agents take turns picking a good according to a specified agent ordering until there are no more goods left.
For brevity, we will say that a picking sequence satisfies a fairness notion if the allocation that it returns always satisfies that notion.
With additive valuations, \citet{ChakrabortySeSu22} showed that for any $x\in [0,1]$, a picking sequence that assigns each subsequent pick to an agent~$i\in N$ with the smallest ratio $\frac{t_i + (1-x)}{w_i}$, where $t_i$ denotes the number of times agent~$i$ has picked so far, satisfies WEF$(x,1-x)$.
Our main result of this section extends their result to submodular valuations.
We make the specification that, in each turn, the agent picks a good that yields the highest \emph{marginal gain} with respect to the agent's current bundle, breaking ties arbitrarily.
More formally, if it is agent $i$'s turn, then $i$ chooses a good $g$ that maximizes $\mgaini{A_i}{g}$, where $A_i$ is the set of goods $i$ has picked in previous turns.

\begin{theorem}
\label{thm:pickseq}
Let $x\in [0,1]$.
Consider a picking sequence~$\pi_x$ such that, in each turn, the pick is assigned to an agent~$i\in N$ with the smallest ratio $\frac{t_i + (1-x)}{w_i}$, where $t_i$ denotes the number of times agent~$i$ has picked so far, and the agent picks a good that yields the highest marginal gain.
Then, under submodular valuations, $\pi_x$ satisfies WMEF$(x,1-x)$.
\end{theorem}

For any $x$ and agents with equal weights, $\pi_x$ encompasses the popular \emph{round-robin algorithm} where the agents take turns in the order $1,2,\dots,n,1,2,\dots,n,1,2,\dots$, and WMEF$(1,0)$ reduces to MEF1 of \citet{CaragiannisKuMo19}.
We therefore have the following corollary, which is also new to the best of our knowledge.

\begin{corollary}
\label{cor:pickseq-unweighted}
Assume that all agents have equal weights and submodular valuations.
Suppose that in each turn of the round-robin algorithm, the picking agent picks a good with the highest marginal gain.
Then, the algorithm returns a complete MEF1 allocation.
\end{corollary}

As Corollary~\ref{cor:pickseq-unweighted} admits a more direct proof, which also illustrates the ideas that we will use to show \Cref{thm:pickseq}, we first present the (shorter) proof of Corollary~\ref{cor:pickseq-unweighted}.

\begin{proof}[Proof of Corollary~\ref{cor:pickseq-unweighted}]
Let $\mathcal{A}$ be the allocation produced by the round-robin algorithm, and consider any $i,j\in N$.
Assume without loss of generality that $i < j$.

We first establish the MEF1 condition from $i$ toward $j$.
Let $k := |A_j| \le |A_i|$, and suppose that agent~$j$ picks the goods in the order $c_1,c_2,\dots,c_k$.
Let $b_1,b_2,\dots,b_k$ be the first $k$ goods picked by agent~$i$ in this order.
For $0\le\ell\le k$, let $B_\ell = \{b_1,\dots,b_\ell\}$ and $C_\ell = \{c_1,\dots,c_\ell\}$ (so $B_0 = C_0 = \emptyset$).
For $1\le\ell\le k$, since agent~$i$ picks $b_\ell$ when $c_\ell$ is also available, it must be that 
\begin{align*}
v_i(B_{\ell}) - v_i(B_{\ell-1}) \ge v_i(B_{\ell-1}\cup \{c_\ell\}) - v_i(B_{\ell-1}).
\end{align*}
Moreover, since $B_{\ell-1} \subseteq A_i \subseteq A_i\cup C_{\ell-1}$, submodularity implies that 
\begin{align*}
v_i(B_{\ell-1}\cup \{c_\ell\}) - v_i(B_{\ell-1}) \ge v_i(A_i\cup C_{\ell-1}\cup\{c_\ell\}) - v_i(A_i\cup C_{\ell-1}).
\end{align*}
Combining the previous two inequalities yields
\begin{align*}
v_i(B_{\ell}) - v_i(B_{\ell-1}) \ge v_i(A_i\cup C_{\ell-1}\cup\{c_\ell\}) - v_i(A_i\cup C_{\ell-1}).
\end{align*}
Summing this over all $\ell\in[k]$, we get $v_i(B_k) \ge v_i(A_i\cup C_k) - v_i(A_i)$.
Since $B_k\subseteq A_i$ and $C_k = A_j$, it follows that $v_i(A_i) \ge v_i(A_i\cup A_j) - v_i(A_i)$, and the MEF1 condition from $i$ to $j$ is fulfilled.

The proof for the MEF1 condition from $j$ toward $i$ is almost identical: by ignoring the first good $g$ picked by agent~$i$ and applying the same argument, we have $v_j(A_j) \ge v_j(A_j\cup (A_i\setminus\{g\})) - v_j(A_j)$.
Thus, the MEF1 condition is again satisfied.
\end{proof}

In \Cref{app:roundrobin-EF1}, we provide an example demonstrating that the condition MEF1 in Corollary~\ref{cor:pickseq-unweighted} cannot be strengthened to EF1, even when agents have matroid-rank valuations.

We now establish \Cref{thm:pickseq} by augmenting the proof of \citet[Thm.~3.2]{ChakrabortySeSu22} from the additive setting with the ideas from our proof of Corollary~\ref{cor:pickseq-unweighted} and arguments involving submodularity.

\begin{proof}[Proof of \Cref{thm:pickseq}]
Fix two agents $i,j \in N$.
For convenience, we write $\pi$ instead of $\pi_x$.
For any prefix $P$ of $\pi$, if $i$ and $j$ pick $t_i$ and $t_j$ times in $P$, respectively, then it must be that $\frac{(t_j+(1-x))-1}{w_j} \le \frac{t_i+(1-x)}{w_i}$; otherwise the $t_j$-th pick of~$j$ should have been assigned to~$i$ instead.
That is, we have $t_i + (1-x) \ge \frac{w_i}{w_j}\cdot(t_j-x)$.
Using this property, we will show that the WMEF$(x,1-x)$ condition from $i$ to $j$ is satisfied after every prefix of $\pi$.

We first prove a general claim that, for any $x,y\in [0,1]$, if the WMEF$(x,y)$ condition from $i$ to $j$ is satisfied with bundles $A_i$ and $A_j$, and we add one good $h\not\in A_i\cup A_j$ to $A_i$, then the condition remains satisfied.
To this end, we show that for every good $g\in A_j$,
\begin{align}
\label{eq:WMEF-addonegood1}
v_i(A_i\cup\{h\}) + y\cdot\mgaini{A_i\cup\{h\}}{g} \ge v_i(A_i) + y\cdot\mgaini{A_i}{g}
\end{align}
and
\begin{align}
\label{eq:WMEF-addonegood2}
v_i(A_i\cup A_j) - v_i(A_i) &- x\cdot\mlossi{A_i\cup A_j}{g} \nonumber \\
&\ge v_i(A_i\cup A_j\cup\{h\}) - v_i(A_i\cup\{h\}) - x\cdot\mlossi{A_i\cup A_j\cup\{h\}}{g}.
\end{align}
From the definition of WMEF$(x,y)$, these two inequalities suffice to prove our claim. 
In order to prove inequality \eqref{eq:WMEF-addonegood1}, we observe that
\begin{align*}
\mgaini{A_i}{h} + y\cdot\mgaini{A_i\cup\{h\}}{g}
&\ge y\cdot \mgaini{A_i}{h} + y\cdot\mgaini{A_i\cup\{h\}}{g} \\
&= y\cdot (v_i(A_i\cup\{h,g\}) - v_i(A_i)) \\
&\ge y\cdot (v_i(A_i\cup\{g\}) - v_i(A_i)) \\
&= y\cdot \mgaini{A_i}{g}.
\end{align*}
Since $\mgaini{A_i}{h} = v_i(A_i\cup\{h\}) - v_i(A_i)$, this implies \eqref{eq:WMEF-addonegood1}.
For \eqref{eq:WMEF-addonegood2}, observe that
\begin{align*}
v_i(A_i\cup A_j) &- v_i(A_i) - x\cdot\mlossi{A_i\cup A_j}{g} \\
&= v_i(A_i\cup A_j\setminus\{g\}) - v_i(A_i) + (1-x)\cdot\mlossi{A_i\cup A_j}{g} \\
&\ge v_i(A_i\cup A_j\cup\{h\}\setminus\{g\}) - v_i(A_i\cup\{h\}) + (1-x)\cdot\mlossi{A_i\cup A_j\cup\{h\}}{g} \\
&= v_i(A_i\cup A_j\cup\{h\}) - v_i(A_i\cup\{h\}) - x\cdot \mlossi{A_i\cup A_j\cup\{h\}}{g},
\end{align*}
where the inequality follows from submodularity.

We are ready to prove \Cref{thm:pickseq}.
Let $\rho = w_i/w_j$ and $y = 1-x$.
From the previous paragraph, it is sufficient to show that the WMEF$(x,1-x)$ condition from $i$ to $j$ is fulfilled after every pick by agent~$j$.
Consider any pick by agent~$j$, and suppose that it is the agent's $t_j$-th pick.
We divide the sequence of picks up to this pick into \emph{phases}, where each phase $\ell\in[t_j]$ consists of the picks after agent~$j$'s $(\ell-1)$-th pick up to (and including) the agent's $\ell$-th pick. 
Let $A_i$ and $A_j$ be the bundle of agent~$i$ and $j$ after phase~$t_j$, respectively.
If $A_j = \emptyset$, then the WMEF$(x,1-x)$ condition from $i$ to $j$ holds trivially, so assume that $A_j \ne \emptyset$.
In addition, we introduce the following notation:
\begin{itemize}
\item $\tau_\ell := $ the number of times agent $i$ picks in phase $\ell$ (that is, between agent $j$'s
 $(\ell-1)$-th and $\ell$-th picks);
\item 
$\alpha_1^\ell, \alpha_2^\ell, \dots,\alpha_{\tau_\ell}^\ell := $
agent $i$'s \emph{marginal gain} for each good that she picks herself in phase $\ell$ with respect to the goods that she has already picked (including those in phase~$\ell$, if any);
\item 
$\alpha_\ell :=  \alpha_1^\ell + \cdots + \alpha_{\tau_\ell}^\ell = $ the total marginal gain of agent $i$ in phase $\ell$ with respect to the goods that she has picked in prior phases;
\item 
$\beta_\ell := $ agent $i$'s marginal gain for the good that agent $j$ picks at the end of phase $\ell$ \emph{with respect to~$A_i$}.
\end{itemize}
Note that $\alpha_1^\ell, \alpha_2^\ell, \dots,\alpha_{\tau_\ell}^\ell$ (and therefore $\alpha_\ell$) and $\beta_\ell$ are defined differently than in the proof of Theorem~3.2 of \citet{ChakrabortySeSu22}, as the valuations that we consider may be non-additive.

For any integer $s\in[t_j]$,
applying the condition in the first paragraph of our proof to the picking sequence up to and including phase $s$, we have
\begin{equation}
\label{eq:picking-picks}
y+\sum_{\ell=1}^s \tau_\ell ~\ge~ \rho(s-x)
\qquad\qquad \forall s\in [t_j].
\end{equation}
Every time it is agent $i$'s turn, she picks a good with the highest marginal gain with respect to her current bundle among the available goods. 
In particular, in each phase $\ell$, she picks $\tau_{\ell}$ goods each of which yields at least as high marginal gain to her as any good not yet picked by agent $j$. 
By submodularity, this implies
\begin{equation}
\label{eq:picking-utils}
\alpha_\ell
\geq \tau_\ell \cdot \max_{\ell\le r \le t_j} \beta_r
\qquad\qquad \forall \ell\in [t_j].
\end{equation}
Using \eqref{eq:picking-picks} and \eqref{eq:picking-utils}, the same inductive argument as in Chakraborty et al.'s proof of their Theorem~3.2 yields
\begin{align}
\label{eq:picking-summary}
y\cdot \max_{1\le r\le t_j}\beta_r + \sum_{\ell=1}^{t_j}
\alpha_{\ell}
~\ge~ \rho\left(\sum_{\ell=1}^{t_j}\beta_\ell - x\beta_1 \right).
\end{align}

Let $g^* \in \argmax_{g\in A_j}\mgaini{A_i}{g}$, and let $g_1$ be the first good picked by agent~$j$ (possibly $g_1 = g^*$).
Using \eqref{eq:picking-summary}, we get
\begin{align*}
(1-x)\cdot\mgaini{A_i}{g^*} + v_i(A_i) 
&\ge \frac{w_i}{w_j}\cdot\left(\sum_{g\in A_j}\mgaini{A_i}{g} - x\cdot\mgaini{A_i}{g_1}\right) \\
&\ge \frac{w_i}{w_j}\cdot\left(\sum_{g\in A_j}\mgaini{A_i}{g} - x\cdot\mgaini{A_i}{g^*}\right) \\
&= \frac{w_i}{w_j}\cdot\left(\sum_{g\in A_j\setminus\{g^*\}}\mgaini{A_i}{g} + (1-x)\cdot\mgaini{A_i}{g^*}\right) \\
&\ge \frac{w_i}{w_j}\cdot\left(v_i(A_i\cup A_j\setminus\{g^*\}) - v_i(A_i) + (1-x)\cdot\mgaini{A_i}{g^*}\right),
\end{align*}
where the second inequality follows from the definition of~$g^*$ and the last inequality from submodularity.
Consequently, we have
\begin{align*}
&\frac{v_i(A_i) + (1-x)\cdot\mgaini{A_i}{g^*}}{w_i} \\
&\ge \frac{v_i(A_i\cup A_j\setminus\{g^*\}) - v_i(A_i) + (1-x)\cdot\mgaini{A_i}{g^*}}{w_j} \\
&\ge \frac{v_i(A_i\cup A_j\setminus\{g^*\}) - v_i(A_i) + (1-x)\cdot\mgaini{A_i\cup A_j\setminus\{g^*\}}{g^*}}{w_j} \\
&= \frac{v_i(A_i\cup A_j) - v_i(A_i) - x\cdot\mlossi{A_i\cup A_j}{g^*}}{w_j}.
\end{align*}
Here, the second inequality holds by submodularity.
As a result, the WMEF$(x,1-x)$ condition between agents~$i$ and $j$ is fulfilled, completing the proof.
\end{proof}

\section{Nash Welfare}
\label{sec:Nash}

In this section, we turn our attention to maximum weighted Nash welfare (MWNW), a weighted extension of the well-studied maximum Nash welfare (MNW).
MWNW has been studied in several recent papers \citep{ChakrabortyIgSu21,ChakrabortyScSu21,ChakrabortySeSu22,GargKuKu20,GargHuVe21,GargHuMu22,SuksompongTe22,ViswanathanZi22}.

\begin{definition}[MWNW]
Given an instance, an allocation $\mathcal{A}$ is a \emph{maximum weighted Nash welfare (MWNW)} allocation if it maximizes the \emph{weighted Nash welfare} $\text{WNW}(\mathcal{A}) := \prod_{i\in N}v_i(A_i)^{w_i}$.
If the highest possible weighted Nash welfare is $0$, an MWNW allocation should maximize the number of agents receiving positive utility and, subject to that, maximize the weighted Nash welfare of these agents. 
\end{definition}

\citet{ChakrabortySeSu22} showed that, for any $x\in [0,1]$, there exists an instance with binary additive valuations and identical goods such that every MWNW allocation is not WEF$(x,1-x)$.
As a consequence, MWNW allocations cannot always satisfy WMEF$(x,1-x)$ for submodular valuations.
On the other hand, \citet{ChakrabortyIgSu21} proved that, under additive valuations, MWNW allocations satisfy \emph{weak WEF1 (WWEF1)}, which is weaker than WEF$(x,1-x)$ for every $x$ but still reduces to EF1 in the unweighted additive setting.
We extend their result to the submodular setting via a natural generalization of WWEF1.

\begin{definition}[WWMEF1]
An allocation $\mathcal{A}$ is said to satisfy \emph{weak weighted marginal envy-freeness up to one good (WWMEF1)} if for any pair of agents $i,j$ with $A_j\neq\emptyset$, there exists a good $g\in A_j$ such that
\[
\text{either } 
\frac{v_i(A_i)}{w_i} \geq \frac{v_i(A_i \cup A_j \setminus \{g\}) - v_i(A_i)}{w_j}
\text{ or }
\frac{v_i(A_i \cup \{g\})}{w_i} \geq \frac{v_i(A_i \cup A_j) - v_i(A_i)}{w_j}.
\]
\end{definition}

\begin{theorem}
 Under submodular valuations, any MWNW allocation satisfies WWMEF1 and PO.
\end{theorem}

\begin{proof}
    Let $\mathcal{A}$ be an MWNW allocation, and recall that $N^+_\mathcal{A}$ denotes the set of agents who receive positive utility from~$\mathcal{A}$.
    We first prove the PO property.  
    If $\mathcal{A}$ were not PO, there would exist an allocation $\widehat{\mathcal{A}}$ such that $v_j(\widehat{A}_j) > v_j(A_j)$ for some $j \in N$ and $v_i(\widehat{A}_i) \geq v_i(A_i)$ for every $i \in N \setminus \{j\}$. 
    If $j \in N \setminus N^+_\mathcal{A}$, we would have $v_i(\widehat{A}_i) > 0$ for every $i \in N^+_\mathcal{A} \cup \{j\}$, contradicting the assumption that $N^+_\mathcal{A}$ is a largest subset of agents to whom it is possible to give positive utility simultaneously. 
    On the other hand, if $j \in N^+_\mathcal{A}$, we would have $\prod_{i \in N^+_\mathcal{A}} v_i(\widehat{A}_i)^{w_i} > \prod_{i \in N^+_\mathcal{A}} v_i(A_i)^{w_i}$, which would mean that $\mathcal{A}$ does not maximize the weighted Nash welfare of the agents in $N^+_\mathcal{A}$, again a contradiction.
    Therefore, $\mathcal{A}$ is PO.

    Next, we proceed to establish the WWMEF1 property. 
    Following the approach of \citet{CaragiannisKuMo19} and \citet{ChakrabortyIgSu21}, we first prove that $\mathcal{A}$ is WWMEF1 for the scenario $N^+_\mathcal{A} = N$ and then address the case where $N^+_\mathcal{A} \neq N$. 
    
    Assume that $N^+_\mathcal{A} = N$, and suppose for contradiction that $\mathcal{A}$ is not WWMEF1.
    This means that there exists a pair of agents $i, j \in N$ such that the WWMEF1 condition between $i$ and $j$ is violated. 
    That is, $A_j\ne\emptyset$, and for every good $g\in A_j$ it holds that
    \begin{equation}
    \label{eq:WWMEF1-1}
    \frac{v_i(A_i)}{w_i} < \frac{v_i(A_i \cup A_j \setminus \{g\}) - v_i(A_i)}{w_j}    
    \end{equation}
    and
    \begin{equation}
    \label{eq:WWMEF1-2}
    \frac{v_i(A_i \cup \{g\})}{w_i} < \frac{v_i(A_i \cup A_j) - v_i(A_i)}{w_j}. 
    \end{equation}
    We will construct another allocation $\mathcal{A}'$ obtained by transferring a good $g^*$ (to be chosen later) from $A_j$ to~$A_i$, so that $A'_i = A_i \cup \{g^*\}$ and $A'_j = A_j \setminus \{g^*\}$, and the bundles of all other agents remain unchanged. 
    Then, we have
    \begin{align}
        \frac{\text{WNW}(\mathcal{A}')}{\text{WNW}(\mathcal{A})} 
        & = \left( \frac{v_i(A_i \cup \{ g^*\})}{v_i(A_i)} \right)^{w_i} \left( \frac{v_j(A_j \setminus \{g^*\})}{v_j(A_j)} \right)^{w_j} \nonumber \\
        & = \left( \frac{v_i(A_i) + \mgaini{A_i}{g^*}}{v_i(A_i)} \right)^{w_i} \left( \frac{v_j(A_j) - \mloss{j}{A_j}{g^*}}{v_j(A_j)} \right)^{w_j} \nonumber  \\
        & = \frac{(v_i(A_i) + \mgaini{A_i}{g^*})^{w_i} (v_j(A_j) - \mloss{j}{A_j}{g^*})^{w_j} }{v_i(A_i)^{w_i} v_j(A_j)^{w_j}}. \label{thm:mwnw_wwmef1-wnwratio}
    \end{align}
    
    By Lemma C.5 of \citet{CaragiannisKuMo19} (which is a simple application of submodularity), it holds that 
    \begin{equation}
    \label{eq:remove-one}
    \sum_{g \in A_j} \mloss{j}{A_j}{g} \leq v_j(A_j).
    \end{equation}
    Moreover, by submodularity, we have that for every $g'\in A_j$,
    \begin{equation}
    \label{eq:submod-1}
    \sum_{g \in A_j} \mgaini{A_i}{g} 
    \geq v_i(A_i \cup A_j \setminus \{g'\}) - v_i(A_i) + \mgaini{A_i}{g'},
    \end{equation}
    and 
    \begin{equation}
        \label{eq:submod-2}
            \sum_{g \in A_j} \mgaini{A_i}{g} \geq v_i(A_i \cup A_j) - v_i(A_i) > 0,
    \end{equation}
    where the last inequality follows from \eqref{eq:WWMEF1-2}.
    By \eqref{eq:submod-2}, we have that $\mgaini{A_i}{g} > 0$ for at least one $g\in A_j$. 

    Let $B = \{g \in A_j\colon \mgaini{A_i}{g} > 0\}$, and let
    $g^* \in \argmin_{g \in B}\frac{\mloss{j}{A_j}{g}}{\mgaini{A_i}{g}}$. 
    Due to our choice of~$g^*$ and the definition of~$B$, we have
    \begin{equation} \label{thm:mwnw_wwmef1-deltaratio}
        \frac{\mloss{j}{A_j}{g^*}}{\mgaini{A_i}{g^*}} 
        \leq \frac{\sum_{g \in B} \mloss{j}{A_j}{g}}{\sum_{g \in B} \mgaini{A_i}{g}} 
        \leq \frac{\sum_{g \in A_j} \mloss{j}{A_j}{g}}{\sum_{g \in A_j} \mgaini{A_i}{g}}.
    \end{equation}
    Note that $v_i(A_i \cup A_j \setminus \{g^*\}) - v_i(A_i) + \mgaini{A_i}{g^*} \ge \mgaini{A_i}{g^*} > 0$.
    Combining (\ref{thm:mwnw_wwmef1-deltaratio}) with \eqref{eq:remove-one}, \eqref{eq:submod-1}, and \eqref{eq:submod-2}, we get
    \begin{equation} \label{thm:mwnw_wwmef1-wnwratio_2} 
        \frac{\mloss{j}{A_j}{g^*}}{\mgaini{A_i}{g^*}} \leq \frac{v_j(A_j)}{v_i(A_i \cup A_j \setminus \{g^*\}) - v_i(A_i) + \mgaini{A_i}{g^*}}
    \end{equation}
    and
    \begin{equation} \label{thm:mwnw_wwmef1-wnwratio_3}
        \frac{\mloss{j}{A_j}{g^*}}{\mgaini{A_i}{g^*}} \leq \frac{v_j(A_j)}{v_i(A_i \cup A_j) - v_i(A_i)}.
    \end{equation}
    
    We split our remaining argument into two cases.
    In each case, we will show that the expression in~(\ref{thm:mwnw_wwmef1-wnwratio}) is greater than $1$, which results in a contradiction because $\mathcal{A}$ is an MWNW allocation.
    
    \paragraph{Case 1: $w_i \geq w_j$.}
    Assume first that $\mloss{j}{A_j}{g^*} > 0$.
    By \eqref{eq:WWMEF1-1} and (\ref{thm:mwnw_wwmef1-wnwratio_2}), we have
    \begin{align*}
            \frac{v_i(A_i)}{w_i} & < \frac{v_i(A_i \cup A_j \setminus \{g^*\})-v_i(A_i)}{w_j} 
            \leq \frac{\frac{\mgaini{A_i}{g^*}}{\mloss{j}{A_j}{g^*}} \cdot v_j(A_j) - \mgaini{A_i}{g^*}}{w_j}.
    \end{align*} 
    Multiplying $\mloss{j}{A_j}{g^*} \cdot w_i$ on both sides and manipulating, we get
    \begin{equation*}
        \frac{w_i}{w_j} \cdot \mgaini{A_i}{g^*} \cdot v_j(A_j) - \frac{w_i}{w_j}\cdot \mgaini{A_i}{g^*} \cdot \mloss{j}{A_j}{g^*} - \mloss{j}{A_j}{g^*} \cdot v_i(A_i) > 0.
    \end{equation*}
    Since $N^+_\mathcal{A} = N$, we can divide by $v_i(A_i)v_j(A_j)$ on both sides to obtain
    \begin{equation*}
        \frac{w_i}{w_j} \left( \frac{\mgaini{A_i}{g^*}}{v_i(A_i)} - \frac{\mgaini{A_i}{g^*} \mloss{j}{A_j}{g^*}}{v_i(A_i)v_j(A_j)} \right) - \frac{\mloss{j}{A_j}{g^*}}{v_j(A_j)} > 0.
    \end{equation*}
    Adding $1$ to both sides yields
    \begin{equation*}
        1+ \frac{w_i}{w_j} \left( \frac{\mgaini{A_i}{g^*}}{v_i(A_i)} - \frac{\mgaini{A_i}{g^*} \mloss{j}{A_j}{g^*}}{v_i(A_i)v_j(A_j)} \right) - \frac{\mloss{j}{A_j}{g^*}}{v_j(A_j)} > 1,
    \end{equation*}
    which can then be factorized as
    \begin{equation}
    \label{eq:case1-factor}
        \left( 1 + \frac{w_i}{w_j}\cdot \frac{\mgaini{A_i}{g^*}}{v_i(A_i)} \right) \biggl( 1 - \frac{\mloss{j}{A_j}{g^*}}{v_j(A_j)} \biggr) > 1.
    \end{equation}
    If $\mloss{j}{A_j}{g^*} = 0$, then \eqref{eq:case1-factor} holds trivially because $\mgaini{A_i}{g^*} > 0$.
    Hence, \eqref{eq:case1-factor} always holds.
    
    Now, since $\frac{\mgaini{A_i}{g^*}}{v_i(A_i)} > 0$ and $\frac{w_i}{w_j} \geq 1$, Bernoulli's inequality together with \eqref{eq:case1-factor} implies that
    \begin{equation*}
        \left( 1 + \frac{\mgaini{A_i}{g^*}}{v_i(A_i)} \right)^{\frac{w_i}{w_j}} \biggl( 1 - \frac{\mloss{j}{A_j}{g^*}}{v_j(A_j)} \biggr) 
        \geq \left( 1 + \frac{w_i}{w_j}\cdot \frac{\mgaini{A_i}{g^*}}{v_i(A_i)} \right) \biggl( 1 - \frac{\mloss{j}{A_j}{g^*}}{v_j(A_j)} \biggr) > 1,
    \end{equation*}
    which means that
    \begin{equation*}
        \left( \frac{(v_i(A_i) + \mgaini{A_i}{g^*})^{w_i} (v_j(A_j)-\mloss{j}{A_j}{g^*})^{w_j}}{v_i(A_i)^{w_i} v_j(A_j)^{w_j}} \right)^{\frac{1}{w_j}} 
        > 1,
    \end{equation*}
    that is (from (\ref{thm:mwnw_wwmef1-wnwratio})),
    \begin{equation*}
        \left[ \frac{\text{WNW}(\mathcal{A}')}{\text{WNW}(\mathcal{A})} \right]^{\frac{1}{w_j}} > 1,
    \end{equation*}
    a contradiction.

    \paragraph{Case 2: $w_i < w_j$.}
    Assume first that $\mloss{j}{A_j}{g^*} > 0$.
    By \eqref{eq:WWMEF1-2} and (\ref{thm:mwnw_wwmef1-wnwratio_3}), we have
    \begin{align*}
    \frac{w_j}{w_i} \cdot (v_i(A_i) + \mgaini{A_i}{g^*}) &= \frac{w_j}{w_i} \cdot v_i(A_i \cup \{g^*\}) \\
    &< v_i(A_i \cup A_j )-v_i(A_i) 
    \leq \frac{\mgaini{A_i}{g^*}}{\mloss{j}{A_j}{g^*}}\cdot v_j(A_j).
    \end{align*}
    Multiplying $\mloss{j}{A_j}{g^*}$ on both sides and manipulating, we get
    \begin{equation*}
        \mgaini{A_i}{g^*} \cdot v_j(A_j) - \frac{w_j}{w_i} \cdot \mloss{j}{A_j}{g^*} \cdot v_i(A_i) - \frac{w_j}{w_i} \cdot \mgaini{A_i}{g^*} \cdot \mloss{j}{A_j}{g^*}> 0.
    \end{equation*}
    Since $N^+_\mathcal{A} = N$, we can divide by $v_i(A_i)v_j(A_j)$ on both sides to obtain
    \begin{equation*}
        \frac{\mgaini{A_i}{g^*}}{v_i(A_i)} - \frac{w_j}{w_i} \left( \frac{\mloss{j}{A_j}{g^*}}{v_j(A_j)} + \frac{\mgaini{A_i}{g^*} \mloss{j}{A_j}{g^*}}{v_i(A_i)v_j(A_j)} \right) > 0.
    \end{equation*}
    Adding $1$ to both sides yields
    \begin{equation*}
        1+ \frac{\mgaini{A_i}{g^*}}{v_i(A_i)} - \frac{w_j}{w_i} \left( \frac{\mloss{j}{A_j}{g^*}}{v_j(A_j)} + \frac{\mgaini{A_i}{g^*} \mloss{j}{A_j}{g^*}}{v_i(A_i)v_j(A_j)} \right) > 1,
    \end{equation*}
    which can then be factorized as
    \begin{equation}
    \label{eq:case2-factor}
        \left( 1 + \frac{\mgaini{A_i}{g^*}}{v_i(A_i)} \right) \biggl( 1 - \frac{w_j}{w_i}\cdot \frac{\mloss{j}{A_j}{g^*}}{v_j(A_j)} \biggr) > 1.
    \end{equation}
    If $\mloss{j}{A_j}{g^*} = 0$, then \eqref{eq:case2-factor} holds trivially because $\mgaini{A_i}{g^*} > 0$.
    Hence, \eqref{eq:case2-factor} always holds.
    Moreover, since $w_j > w_i$, it must be that $\mloss{j}{A_j}{g^*} < v_j(A_j)$.
    
    Now, since $\frac{\mloss{j}{A_j}{g^*}}{v_j(A_j)} < 1$ and $\frac{w_j}{w_i} > 1$, Bernoulli's inequality together with \eqref{eq:case2-factor} implies that
    \begin{equation*}
        \left( 1 + \frac{\mgaini{A_i}{g^*}}{v_i(A_i)} \right) \biggl( 1 - \frac{\mloss{j}{A_j}{g^*}}{v_j(A_j)} \biggr)^{\frac{w_j}{w_i}} \geq \left( 1 + \frac{\mgaini{A_i}{g^*}}{v_i(A_i)} \right) \biggl( 1 - \frac{w_j}{w_i}\cdot \frac{\mloss{j}{A_j}{g^*}}{v_j(A_j)} \biggr) > 1,
    \end{equation*}
    which means that
    \begin{equation*}
        \left( \frac{(v_i(A_i) + \mgaini{A_i}{g^*})^{w_i} (v_j(A_j)-\mloss{j}{A_j}{g^*})^{w_j}}{v_i(A_i)^{w_i} v_j(A_j)^{w_j}} \right)^{\frac{1}{w_i}} 
        > 1,
    \end{equation*}
    that is (from (\ref{thm:mwnw_wwmef1-wnwratio})),
    \begin{equation*}
        \left[ \frac{\text{WNW}(\mathcal{A}')}{\text{WNW}(\mathcal{A})} \right]^{\frac{1}{w_i}} > 1,
    \end{equation*}
    a contradiction.
    This completes the proof for the scenario where $N^+_\mathcal{A} = N$. 
    
    Finally, we handle the scenario where $N^+_\mathcal{A} \subsetneq N$. 
    Let $i,j\in N$ with $A_j\neq\emptyset$, and consider three cases.
    \begin{itemize}
    \item If $i,j\in N^+_\mathcal{A}$, we can show as in the scenario where $N^+_\mathcal{A} = N$ that the WWMEF1 condition between $i$ and $j$ is satisfied.
    \item Suppose that $i\not\in N^+_\mathcal{A}$ and $j\in N^+_\mathcal{A}$.
    This means that $v_i(A_i) = 0$ and $v_j(A_j) > 0$.
    If there exists a good $g\in A_j$ such that $v_i(A_i\cup A_j\setminus\{g\}) = 0$, the WWMEF1 condition from $i$ to $j$ is trivially satisfied.
    Thus, assume now  that $v_i(A_i\cup A_j\setminus\{g\}) > 0$ for every $g\in A_j$.
    Fix an arbitrary $\widehat{g}\in A_j$.
    Since $v_i(A_i\cup A_j\setminus\{\widehat{g}\}) > 0$, by submodularity, there exists $g'\in A_j$ such that $v_i(A_i\cup \{g'\}) > 0$.
    Similarly, since $v_i(A_i\cup A_j\setminus\{g'\}) > 0$, there exists $g''\in A_j$ with $g''\neq g'$ (possibly $g'' = \widehat{g}$) such that $v_i(A_i\cup\{g''\}) > 0$.
    
    If $v_j(A_j\setminus\{g'\}) > 0$, we can transfer $g'$ from $A_j$ to $A_i$ and obtain an allocation with more agents receiving positive utility than in $\mathcal{A}$, a contradiction.
    Therefore, $v_j(A_j\setminus\{g'\}) = 0$.
    Similarly, $v_j(A_j\setminus\{g''\}) = 0$.
    By submodularity, we must have $v_j(A_j) = 0$ as well, contradicting the assumption that $j\in N^+_\mathcal{A}$.
    
    \item Suppose that $j\not\in N^+_\mathcal{A}$.
    This means that $v_j(A_j) = 0$.
    If $v_i(A_i\cup\{g\}) > v_i(A_i)$ for some $g\in A_j$, we can transfer $g$ from $A_j$ to $A_i$ and obtain an allocation with the same number of agents receiving positive utility as in $\mathcal{A}$ but a higher weighted Nash welfare of these agents than in $\mathcal{A}$, a contradiction.
    Hence, $v_i(A_i\cup\{g\}) = v_i(A_i)$ for every $g\in A_j$.
    Submodularity implies that $v_i(A_i\cup A_j) = v_i(A_i)$.
    Therefore, the WWMEF1 condition from $i$ to $j$ is satisfied.
    \end{itemize}
    It follows that $\mathcal{A}$ is WWMEF1 in all cases.
\end{proof}

\citet{ViswanathanZi22} showed that if agents have matroid-rank valuations, then an MWNW allocation can be found in polynomial time.
On the other hand, with equal-weight agents and additive valuations, even approximating the maximum Nash welfare is computationally difficult \citep{Lee17}.

\section{Transfer Algorithm}
\label{sec:transfer}

For agents with equal weights and matroid-rank valuations, \citet[Algorithm~1]{BenabbouChIg21} proposed a ``transfer algorithm'' that computes a clean, utilitarian welfare-maximizing EF1 allocation  in polynomial time.
In this section, we extend their algorithm to the weighted setting.
Our algorithm is presented as \Cref{alg:TWEF1-transfer}; we argue in the proof of \Cref{thm:transfer-algo} that the algorithm is well-defined.
\begin{algorithm}[!ht]
    \caption{For finding a clean TWEF$(x, 1-x)$ allocation maximizing $\sum_{i \in N} v_i(A_i)$}\label{alg:TWEF1-transfer}
    \begin{algorithmic}
        \State Compute a clean allocation $\mathcal{A}$ that maximizes the unweighted utilitarian welfare.
        \While{there exist $i, j \in N$ such that TWEF$(x, 1-x)$ from $i$ to $j$ fails with respect to $\mathcal{A}$}
            \State Find a good $g \in A_j$ with $\mgaini{A_i}{g} = 1$.
            \State $A_i \gets A_i\cup\{g\}$; $A_j \gets A_j\setminus\{g\}$.
        \EndWhile
    \end{algorithmic}
\end{algorithm}

\begin{theorem}
\label{thm:transfer-algo}
Suppose that all agents have matroid-rank valuations, and let $x\in [0,1]$.
\Cref{alg:TWEF1-transfer} with parameter~$x$ returns a clean TWEF$(x,1-x)$ (and therefore WMEF$(x,1-x)$) allocation that maximizes the unweighted utilitarian welfare among all allocations in polynomial time.
\end{theorem}

Since any allocation maximizing the unweighted utilitarian welfare is PO, the allocation output by \Cref{alg:TWEF1-transfer} is also PO.
In the unweighted setting, \citet{BenabbouChIg21} exhibited polynomial-time termination of their algorithm using the potential function $\Phi(\mathcal{A}) := \sum_{i\in N}v_i(A_i)^2$.
As $\Phi(\mathcal{A})$ is always an integer between $0$ and $m^2$ and decreases with every transfer, the number of transfers made by their algorithm is at most $m^2$.
While we can also establish termination of our weighted algorithm by modifying the potential function to $\Phi(\mathcal{A}) = \sum_{i\in N}\frac{v_i(A_i)^2 + (1-2x)\cdot v_i(A_i)}{w_i}$, this argument does not yield a polynomial upper bound on the number of transfers, because the potential function may decrease by a very small amount depending on the weights.
Therefore, we will instead employ a different, more refined, argument to show that our algorithm terminates in polynomial time as well. 

\begin{proof}[Proof of \Cref{thm:transfer-algo}]
By Proposition~\ref{prop:TWEF-WMEF}, it suffices to prove the statement for TWEF$(x,1-x)$.

First, we claim that each transfer maintains the welfare optimality and cleanness of the allocation.
Indeed, $v_j(A_j)$ decreases by $1$ because the previous allocation is clean, while $v_i(A_i)$ increases by $1$ due to the algorithm's choice of the good $g\in A_j$.
Hence, $\sum_{i\in N}v_i(A_i)$ remains the same.
Moreover, since $v_i(A_i) = |A_i|$ for all $i\in N$, the allocation remains clean.

If the TWEF$(x,1-x)$ condition from agent~$i$ to agent~$j$ fails at some point during the execution of the algorithm, it must be that $v_i(A_i) < v_i(A_i\cup A_j)$, and for every $g\in A_j$ we have
\begin{align}
\frac{v_i(A_i) + (1-x)\cdot \mgaini{A_i}{g}}{w_i} 
&< \frac{v_i(A_j) - x\cdot \mlossi{A_j}{g}}{w_j} \nonumber \\
&= \frac{v_i(A_j\setminus\{g\}) + (1-x)\cdot \mlossi{A_j}{g}}{w_j} \nonumber \\
&\le \frac{v_j(A_j\setminus\{g\}) + (1-x)\cdot \mloss{j}{A_j}{g}}{w_j}, \label{eq:transfer-TWEF}
\end{align}
where the latter inequality follows from cleanness.
Since $v_i(A_i) < v_i(A_i\cup A_j)$, by submodularity, there exists $g^*\in A_j$ such that $\mgaini{A_i}{g^*} = 1$; in particular, the algorithm is well-defined.
Plugging this good~$g^*$ into \eqref{eq:transfer-TWEF} and using the cleanness of $\mathcal{A}$, we get
\begin{align}
\label{eq:transfer-size}
\frac{|A_i| + (1-x)}{w_i} < \frac{|A_j| - x}{w_j}.
\end{align}

If the algorithm terminates, then TWEF$(x,1-x)$ is satisfied for all pairs of agents $i,j$.
We will show that the algorithm always terminates, and moreover does so in polynomial time.
The initial clean allocation $\mathcal{A}$ can be found in polynomial time \citep{BenabbouChIg21}.
Checking whether TWEF$(x,1-x)$ fails for some pair $i,j$ (and, if so, finding a valid transfer) can be done in polynomial time.
It therefore remains to argue that the number of transfers is polynomial.
For ease of understanding, we will formulate this argument in terms of cupboards and balls.

Associate each $i\in N$ with a cupboard consisting of $m$ shelves at height $\frac{1-x}{w_i}, \frac{2-x}{w_i}, \dots, \frac{m-x}{w_i}$, respectively.
For the clean allocation~$\mathcal{A}$ at the beginning of the algorithm and each $i\in N$, place one ball on each of the $|A_i|$ lowest shelves\footnote{Note that the sum of the heights of all balls is $\sum_{i\in N}\frac{|A_i|^2 + (1-2x)\cdot |A_i|}{2w_i}$, which is exactly half of the potential function mentioned before the proof.} of cupboard~$i$.
Whenever a good is transferred from $A_j$ to $A_i$, move the highest ball in cupboard~$j$ to the lowest shelf without a ball in cupboard~$i$.
This means that the ball is moved from height $\frac{|A_j|-x}{w_j}$ to height $\frac{|A_i|+1-x}{w_i}$; by \eqref{eq:transfer-size}, the height of the ball decreases.
Since there are $m$ balls and at most $mn$ heights of the shelves, the number of transfers is at most $m^2n$, which is indeed polynomial.\footnote{Note that if all agents have equal weights, the number of different shelf heights is only $m$. 
The number of transfers is then bounded by $m^2$, which matches the bound provided by \citet{BenabbouChIg21}.}
This concludes the proof.
\end{proof}

\section{Harmonic Welfare}
\label{sec:harmonic}

Recall from \Cref{sec:Nash} that an MWNW allocation maximizes the product $\prod_{i\in N} v_i(A_i)^{w_i}$, or equivalently, the sum $\sum_{i\in N} w_i\cdot \ln v_i(A_i)$.
Since $\ln k$ is approximately the $k$-th harmonic number $H_k := 1 + \frac{1}{2} + \dots + \frac{1}{k}$ for each positive integer $k$, one could also consider a \emph{maximum weighted harmonic welfare (MWHW)} allocation, defined as an allocation maximizing the sum $\sum_{i\in N} w_i\cdot H_{v_i(A_i)}$, where $H_0 = 0$.
Interestingly, we show in this section that for agents with matroid-rank valuations, MWHW outperforms MWNW in terms of fairness.
Specifically, even though for each $x\in [0,1]$ there exists an instance with binary additive valuations and identical goods in which every MWNW allocation fails WEF$(x,1-x)$ \citep{ChakrabortySeSu22}, we show that a clean MWHW allocation satisfies TWEF$(0,1)$ for matroid-rank valuations (and therefore WEF$(0,1)$ for binary additive valuations).
More generally, we define a class of modified harmonic numbers parameterized by $x$ such that a clean maximum-welfare allocation based on each~$x$ satisfies TWEF$(x,1-x)$.

\begin{definition}[Modified harmonic numbers]
    Let $k \in \mathbb{Z}_{\geq 0}$.
    For $x\in [0,1)$, the number $H_{k,x}$ is defined by
    \begin{equation*}
        H_{k,x} =
        \begin{cases} 
          \frac{1}{1-x} + \frac{1}{2-x} + \dots + \frac{1}{k - x} & \text{ if } k \geq 1;\\
          0 & \text{ if } k = 0, 
       \end{cases}
    \end{equation*}
    whereas for $x = 1$, $H_{k,x}$ is defined by
    \begin{equation*}
        H_{k,1} =
        \begin{cases} 
          1 + \frac{1}{2} + \dots + \frac{1}{k-1} & \text{ if } k \geq 2;\\
          0 & \text{ if } k = 1;\\
          -\infty & \text{ if } k = 0.
       \end{cases}
    \end{equation*}
\end{definition}
Note that the numbers $H_{k,0}$ correspond to the canonical harmonic numbers $H_k$, and that for each~$x$, the sequence $H_{0,x}, H_{1,x}, \dots$ is increasing.
We define a maximum weighted harmonic welfare allocation parameterized by $x$.
Recall that $N^+_\mathcal{A}$ denotes the set of agents who receive positive utility from an allocation~$\mathcal{A}$.
\begin{definition}[\mwhw{}]
    For $x \in [0,1)$, given an instance with matroid-rank valuations, an allocation~$\mathcal{A}$ is an \mwhw{} allocation if it maximizes the sum $\text{WHW}_x(\mathcal{A}) := \sum_{i\in N} w_i\cdot H_{v_i(A_i), x}$.
    
    For $x = 1$, $\mathcal{A}$ is an MWHW$_1$ allocation if it maximizes the number of agents receiving positive utility and, subject to that, maximizes the sum $\sum_{i\in N^+_\mathcal{A}} w_i\cdot H_{v_i(A_i),1}$.
\end{definition}

The quantity $H_{v_i(A_i), x}$ is well-defined because, for matroid-rank valuations, $v_i(A_i)$ is always a non-negative integer.
We now prove the efficiency and fairness guarantees of \mwhw{} allocations, starting with efficiency.

\begin{theorem}
\label{thm:harmonic-PO}
Let $x\in [0,1]$.
Under matroid-rank valuations, any \mwhw{} allocation is PO. 
\end{theorem}

\begin{proof}
Let $\mathcal{A}$ be an \mwhw{} allocation.
    For $x < 1$, if $\mathcal{A}$ is Pareto-dominated by another allocation~$\mathcal{A}'$, then $\sum_{i\in N} w_i\cdot H_{v_i(A'_i), x} > \sum_{i\in N} w_i\cdot H_{v_i(A_i), x}$, a contradiction.
    
    Consider the case $x = 1$. 
    If $\mathcal{A}$ were not PO, there would exist an allocation $\widehat{\mathcal{A}}$ such that $v_j(\widehat{A}_j) > v_j(A_j)$ for some $j \in N$ and $v_i(\widehat{A}_i) \geq v_i(A_i)$ for every $i \in N \setminus \{j\}$.
    If $j \in N \setminus N^+_\mathcal{A}$, we would have $v_i(\widehat{A}_i) > 0$ for every $i \in N^+_\mathcal{A} \cup \{j\}$, contradicting the assumption that $N^+_\mathcal{A}$ is the largest subset of agents to whom it is possible to give positive utility simultaneously. 
    On the other hand, if $j \in N^+_\mathcal{A}$, we would have $\sum_{i\in N^+_\mathcal{A}} w_i\cdot H_{v_i(\widehat{A}_i),x} > \sum_{i\in N^+_\mathcal{A}} w_i\cdot H_{v_i(A_i),x}$, again a contradiction.
    Therefore, $\mathcal{A}$ is PO.
\end{proof}

For the fairness guarantee, we will make an assumption that the allocation is clean; we later demonstrate that this assumption is necessary.
We also remark that given any \mwhw{} allocation, one can easily obtain another \mwhw{} allocation in which every agent receives the same utility as before by iteratively removing any good that does not contribute to its owner's utility until no such good exists.

\begin{theorem}
\label{thm:harmonic-TWEF}
Let $x\in [0,1]$.
Under matroid-rank valuations, any clean \mwhw{} allocation satisfies TWEF$(x,1-x)$ (and therefore WMEF$(x,1-x)$). 
\end{theorem}

\begin{proof}
By Proposition~\ref{prop:TWEF-WMEF}, it suffices to prove the statement for TWEF$(x,1-x)$.

Let $\mathcal{A}$ be a clean \mwhw{} allocation.
Assume for contradiction that for some $i,j\in N$, the TWEF$(x,1-x)$ condition from $i$ to $j$ is violated.
This means that $v_i(A_i) < v_i(A_i\cup A_j)$, and for every $g\in A_j$ it holds that
\[
\frac{v_i(A_i) + (1-x)\cdot \mgaini{A_i}{g}}{w_i} < \frac{v_i(A_j) - x\cdot \mlossi{A_j}{g}}{w_j}.
\]
By the same argument as in the proof of \Cref{thm:transfer-algo}, this implies that
\begin{equation} \label{eqn:mwhw_twef_wef}
    \frac{v_i(A_i) + (1-x)}{w_i} < \frac{v_j(A_j) - x}{w_j}.
\end{equation} 
Also, since $v_i(A_i) < v_i(A_i\cup A_j)$, submodularity implies that there exists a good $g^*\in A_j$ such that $\mgaini{A_i}{g^*} = 1$.

We now consider two cases depending on whether $x = 1$.
\paragraph{Case 1: $0 \le x < 1$.} 
If we transfer $g^*$ from $A_j$ to $A_i$, we obtain an allocation~$\mathcal{A}'$ in which $v_i(A_i') = v_i(A_i) + 1$, $v_j(A_j') = v_j(A_j) - 1$, and $v_k(A_k') = v_k(A_k)$ for all $k\in N\setminus\{i,j\}$.
Since $\mathcal{A}$ is an \mwhw{} allocation, it must be that
\begin{align*}
w_i\cdot H_{v_i(A_i), x} + w_j\cdot H_{v_j(A_j), x} \ge w_i\cdot H_{v_i(A_i)+1, x} + w_j\cdot H_{v_j(A_j)-1, x}.
\end{align*}
This is equivalent to
\begin{equation*}
    w_j \cdot \frac{1}{v_j(A_j) - x} - w_i \cdot \frac{1}{v_i(A_i) + 1 - x} \ge 0.
\end{equation*}
Algebraic manipulation gives us
\begin{equation*}
    \frac{v_i(A_i)+1-x}{w_i} \ge \frac{v_j(A_j)-x}{w_j},
\end{equation*}
which contradicts (\ref{eqn:mwhw_twef_wef}).

\paragraph{Case 2: $x = 1$.} 
From (\ref{eqn:mwhw_twef_wef}), we have that
\begin{equation} \label{eqn:mwhw_twef_wef_x=1}
    \frac{v_i(A_i)}{w_i} < \frac{v_j(A_j) - 1}{w_j}.
\end{equation}
Since $v_i(A_i) \ge 0$ and $v_j(A_j)$ is an integer, it must be that $v_j(A_j) \ge 2$.
If $v_i(A_i) = 0$, we can transfer $g^*$ from $A_j$ to $A_i$ and increase the number of agents with positive utility, contradicting the assumption that $\mathcal{A}$ is an MWHW$_1$ allocation. 
Hence, $v_i(A_i) \geq 1$.

The rest of the argument proceeds in a similar way as in Case~1.
If we transfer $g^*$ from $A_j$ to $A_i$, we obtain an allocation~$\mathcal{A}'$ in which $v_i(A_i') = v_i(A_i) + 1$, $v_j(A_j') = v_j(A_j) - 1$, and $v_k(A_k') = v_k(A_k)$ for all $k\in N\setminus\{i,j\}$.
Note that the number of agents with positive utility is the same in $\mathcal{A}$ and $\mathcal{A}'$.
Since $\mathcal{A}$ is an MWHW$_1$ allocation, it must be that
\begin{align*}
w_i\cdot H_{v_i(A_i), 1} + w_j\cdot H_{v_j(A_j), 1} \ge w_i\cdot H_{v_i(A_i)+1, 1} + w_j\cdot H_{v_j(A_j)-1, 1}.
\end{align*}
This is equivalent to
\begin{equation*}
    w_j \cdot \frac{1}{v_j(A_j) - 1} - w_i \cdot \frac{1}{v_i(A_i)} \ge 0.
\end{equation*}
Algebraic manipulation gives us
\begin{equation*}
    \frac{v_i(A_i)}{w_i} \ge \frac{v_j(A_j)-1}{w_j},
\end{equation*}
which contradicts (\ref{eqn:mwhw_twef_wef_x=1}).
\end{proof}

We now exhibit the necessity of the cleanness condition in \Cref{thm:harmonic-TWEF}.

\begin{proposition}
\label{prop:nonclean-TWEF}
There exists an instance and an allocation such that, for every $x\in [0,1]$, the allocation is \mwhw{} but does not satisfy TWEF$(x,1-x)$.
\end{proposition}

\begin{proof}
Consider an instance with $n = 2$ agents whose weights are $w_1 = 1$ and $w_2 = 2$, and $m = 6$ goods.
Agent~$1$ has an additive valuation with value~$1$ for $g_1$ and $0$ for the remaining goods.
Agent~$2$'s valuation~$v_2$ is given by
    \begin{equation*}
        v_2(S) =
        \begin{cases} 
          \min\{3, |S|\} & \text{ if } g_1\not\in S;\\
          \min\{4, |S|\} & \text{ if } g_1\in S,
       \end{cases}
    \end{equation*}
for each bundle $S\subseteq G$.

First, we claim that $v_2$ is matroid-rank. 
The marginal gain from adding~$g_1$ is always~$1$, while the marginal gain from adding any other good is either $0$ or $1$.
To establish submodularity, let $G'\subseteq G''\subseteq G$ and $g\in G\setminus G''$, and assume that $v_2(G'\cup\{g\}) = v_2(G')$; it suffices to prove that $v_2(G''\cup\{g\}) = v_2(G'')$ as well.
From our earlier discussion, it must be that $g\neq g_1$.
If $g_1\in G'$, then $|G''|\ge |G'| \ge 4$ and thus $v_2(G''\cup\{g\}) = v_2(G'')$.
Assume therefore that $g_1\not\in G'$, which means that $|G'| \ge 3$.
If $G'' = G'$, then $v_2(G''\cup\{g\}) = v_2(G'')$ holds trivially.
Otherwise, we have $|G''| \ge |G'| + 1 \ge 4$, and again $v_2(G''\cup\{g\}) = v_2(G'')$.

Fix $x\in [0,1]$.
If $x = 1$, any \mwhw{} allocation must give $g_1$ to agent~$1$, which leaves agent~$2$ with a utility of at most $3$.
Else, for $x < 1$, the maximum weighted harmonic welfare achievable by giving $g_1$ to agent~$1$ is
\[
\frac{1}{1-x} + 2\cdot\left(\frac{1}{1-x}+\frac{1}{2-x}+\frac{1}{3-x}\right),
\]
whereas the maximum achievable by giving $g_1$ to agent~$2$ is
\[ 2\cdot\left(\frac{1}{1-x}+\frac{1}{2-x}+\frac{1}{3-x} + \frac{1}{4-x}\right).
\]
Since $\frac{1}{1-x} = \frac{2}{2-2x} > \frac{2}{4-x}$, every \mwhw{} allocation must again give $g_1$ to agent~$1$.
In particular, for any $x\in [0,1]$, the allocation $\mathcal{A} = (\{g_1,g_2,g_3\}, \{g_4,g_5,g_6\})$ is an \mwhw{} allocation.

To finish the proof, we show that $\mathcal{A}$ violates the TWEF$(x,1-x)$ condition from agent~$2$ toward agent~$1$.
Note that $v_2(A_2) = 3 < 4 = v_2(A_2\cup A_1)$.
Moreover, for any $g\in A_1$, it holds that
\[
\frac{v_2(A_2) + (1-x)\cdot \mgain{2}{A_2}{g}}{w_2} \leq \frac{3 + (1-x)}{2} < 3-x =  \frac{v_2(A_1) - x\cdot \mloss{2}{A_1}{g}}{w_1}.
\]
Hence, the TWEF$(x,1-x)$ condition from agent~$2$ to agent~$1$ is not satisfied.
\end{proof}

By applying results from the recent work of \citet{ViswanathanZi22}, we show in \Cref{app:harmonic-computation} that an \mwhw{} allocation (which additionally maximizes the unweighted utilitarian welfare across all allocations) can be found in polynomial time.

Finally, we remark that it may be interesting to consider harmonic welfare beyond binary valuations.
In \Cref{app:harmonic-more}, we prove that for agents with equal weights and additive valuations, if the value of every agent for every good is an integer (in which case the harmonic welfare is well-defined), then an allocation maximizing the harmonic welfare is always EF1.

\subsection*{Acknowledgments}

This work was partially supported by the Deutsche Forschungsgemeinschaft under grant BR 4744/2-1, by the Singapore Ministry of Education under grant number MOE-T2EP20221-0001, and by an NUS Start-up Grant.

\bibliographystyle{plainnat}
\bibliography{main}

\appendix

\section{Round-Robin Algorithm and EF1}
\label{app:roundrobin-EF1}

In the unweighted setting, it is well-known that if agents have additive valuations, then the round-robin algorithm always produces an EF1 allocation (e.g., \citep[p.~7]{CaragiannisKuMo19}).
A natural question is therefore whether the MEF1 condition in Corollary~\ref{cor:pickseq-unweighted} can be strengthened to EF1.
We show that the answer is negative, even for matroid-rank valuations.

Consider an instance with $m = 8$ goods and $n = 2$ agents with equal weights.
Agent~$1$ has an additive valuation with value~$1$ for each of $g_4$ and $g_8$, and $0$ for the remaining goods.
The value of agent~$2$ for any bundle~$S$ is 
\[
v_2(S) := |S\cap\{g_4\}| + |S\cap\{g_8\}| + \min\{1, |S\cap\{g_1,g_2,g_3\}|\} + \min\{1, |S\cap\{g_5,g_6,g_7\}|\}.
\]
One can check that $v_2$ is matroid-rank.\footnote{In fact, $v_2$ belongs to a subclass of matroid-rank valuations called \emph{$(0,1)$-OXS} \citep{BenabbouChIg21}.}

Assume that the round-robin algorithm lets agent~$1$ pick first, and that the agents break ties in the goods lexicographically.
Under these assumptions, agent~$1$ picks $g_4$, agent~$2$ picks $g_1$, agent~$1$ picks $g_8$, agent~$2$ picks $g_5$, agent~$1$ picks $g_2$, agent~$2$ picks $g_3$, agent~$1$ picks $g_6$, and agent~$2$ picks $g_7$.
Hence, $A_1 = \{g_2,g_4,g_6,g_8\}$ and $A_2 = \{g_1,g_3,g_5,g_7\}$, and so $v_2(A_2) = 2 < 3 = v_2(A_1\setminus\{g\})$ for every $g\in A_1$.
This means that agent~$2$ is not EF1 toward agent~$1$.
Note also that adding the condition $v_i(A_i) = v_i(A_i\cup A_j)$ as in our definition of TWEF$(x,y)$ does not help circumvent this negative result, as we have $v_2(A_2) = 2 < 4 = v_2(A_2\cup A_1)$.

\section{Computing \mwhw{} allocations}
\label{app:harmonic-computation}

Recently, \citet{ViswanathanZi22} introduced a framework for efficiently computing allocations that maximize a range of fairness objectives.
Formally, a \emph{fairness objective} is a function that maps each allocation~$\mathcal{A}$ (more precisely, the utility vector $(v_1(A_1),\dots,v_n(A_n))$ induced by $\mathcal{A}$) to a totally ordered space.
A \emph{gain function} maps each clean allocation and each agent~$i\in N$ to a real number.\footnote{\citet{ViswanathanZi22} also allowed the output of a gain function to be a vector, but we do not need that.}
They showed that if a fairness objective admits a gain function satisfying certain properties, then an allocation that maximizes the objective (along with the unweighted utilitarian welfare) can be computed in polynomial time.

\begin{lemma}[\citet{ViswanathanZi22}]
\label{lem:yankee}
    Suppose that the fairness objective $\Psi$ admits a gain function $\phi$ such that the following conditions are satisfied:
    
    \begin{enumerate}[(i)]
        \item For any two allocations $\mathcal{A}$ and $\mathcal{A}'$, if $\mathcal{A}$ Pareto-dominates $\mathcal{A}'$, then $\Psi(\mathcal{A}) \geq \Psi(\mathcal{A}')$.
        \item For any clean allocation $\mathcal{A}$ and any agent $i\in N$, let $\mathcal{A}^{+i}$ be an allocation resulting from giving a good~$g\not\in \bigcup_{k\in N}A_k$ such that $\mgaini{A_i}{g} = 1$ to $i$.
        Define $\mathcal{A}^{+j}$ analogously.
        If $\phi(\mathcal{A}, i) \geq \phi(\mathcal{A}, j)$, then $\Psi(\mathcal{A}^{+i}) \geq \Psi(\mathcal{A}^{+j})$; equality holds if and only if $\phi(\mathcal{A}, i) = \phi(\mathcal{A}, j)$.
        \item For any clean allocations $\mathcal{A}$ and $\mathcal{A}'$, if $|A_i| \leq |A'_i|$, then $\phi(\mathcal{A},i) \geq \phi(\mathcal{A}',i)$; equality holds if and only if $|A_i| = |A'_i|$.
        \item The gain function $\phi$ can be computed in polynomial time.
    \end{enumerate}
    Then, there exists a polynomial-time algorithm that computes an allocation that maximizes the fairness objective $\Psi$ as well as the unweighted utilitarian welfare.
\end{lemma}

We now apply Lemma~\ref{lem:yankee} to \mwhw{}.

\begin{theorem}
\label{thm:harmonic-computation}
    For any $x\in [0,1]$, under matroid-rank valuations, an \mwhw{} allocation that maximizes the unweighted utilitarian welfare can be computed in polynomial time.
\end{theorem}

\begin{proof}
Fix $x\in [0,1]$, and let 
\[
\Psi_x(\mathcal{A}) = 
\begin{cases} 
    \sum_{i \in N} w_i \cdot H_{v_i(A_i),x} & \text{ if } x < 1;\\
    \left(|N^+_\mathcal{A}|, \sum_{i \in N^+_\mathcal{A}} w_i \cdot H_{v_i(A_i),1}\right) & \text{ if } x = 1, 
\end{cases}
\]
where for $x = 1$ we compare the ordered pairs $\Psi(\mathcal{A})$ for different allocations~$\mathcal{A}$ lexicographically.
By definition of \mwhw{}, an allocation~$\mathcal{A}$ that maximizes $\Psi_x(\mathcal{A})$ is also an \mwhw{} allocation.
It is clear that $\Psi_x$ satisfies condition~(i) of Lemma~\ref{lem:yankee}.

Next, let $w_{\max} = \max_{i\in N} w_i$.
We define the gain function $\phi_x$ as follows:
\begin{equation*}
        \phi_x(\mathcal{A},i) = 
        \begin{cases}
            \frac{w_i}{|A_i|+1-x} & \text{ if } |A_i| > 0 \text{ or } x < 1;\\
            w_{\max} + 1 & \text{ if } |A_i| = 0 \text{ and } x = 1.
        \end{cases}
\end{equation*}
Since $\phi_x$ can be computed efficiently, condition~(iv) of Lemma~\ref{lem:yankee} is satisfied.
Condition~(iii) is also trivially met.

It remains to show that condition~(ii) is satisfied.
Let $\mathcal{A}$ be a clean allocation, so $v_i(A_i) = |A_i|$ for all~$i\in N$.
Consider any $i,j\in N$.
First, suppose that $x < 1$.
We have 
\begin{align*}
\phi_x(\mathcal{A}, i) \geq \phi_x(\mathcal{A}, j)
&\Longleftrightarrow \frac{w_i}{|A_i|+1-x} \geq \frac{w_j}{|A_j|+1-x} \\
&\Longleftrightarrow w_i\cdot (H_{v_i(A_i)+1,x} - H_{v_i(A_i),x}) \ge w_j\cdot (H_{v_j(A_j)+1,x} - H_{v_j(A_j),x}) \\
&\Longleftrightarrow \Psi_x(\mathcal{A}^{+i}) \geq \Psi_x(\mathcal{A}^{+j}).
\end{align*}
This reasoning also shows that $\phi_x(\mathcal{A}, i) = \phi_x(\mathcal{A}, j)$ if and only if $\Psi_x(\mathcal{A}^{+i}) = \Psi_x(\mathcal{A}^{+j})$.
Hence, (ii) is satisfied.

Next, suppose that $x = 1$.
We consider four cases.

\paragraph{Case 1: $|A_i| = |A_j| = 0$.} 
We have $\phi_1(\mathcal{A}, i) = \phi_1(\mathcal{A}, j) = w_{\text{max}}+1$.
Also, 
\[
\Psi_1(\mathcal{A}^{+i}) = \Psi_1(\mathcal{A}^{+j}) = \left(|N^+_\mathcal{A}| + 1, \sum_{k \in N^+_\mathcal{A}} w_k \cdot H_{v_k(A_k),1}\right),
\]
where the second coordinate is the same as in $\Psi_1(\mathcal{A})$ because $H_{1,1} = 0$.
Hence, (ii) holds.

\paragraph{Case 2: $|A_i| > |A_j| = 0$.} 
We have $\phi_1(\mathcal{A}, j) = w_{\text{max}} + 1 > w_i \ge \phi_1(\mathcal{A}, i)$, so (ii) holds trivially since the assumption $\phi_1(\mathcal{A}, i) \ge \phi_1(\mathcal{A}, j)$ is not satisfied.

\paragraph{Case 3: $|A_j| > |A_i| = 0$.} 
We have $\phi_1(\mathcal{A}, i) = w_{\text{max}} + 1 > w_j \ge \phi_1(\mathcal{A}, j)$.
Also, the first coordinate of $\Psi_1(\mathcal{A}^{+i})$ is $|N^+_\mathcal{A}| + 1$, whereas that of $\Psi_1(\mathcal{A}^{+j})$ is $|N^+_\mathcal{A}|$, which means that $\Psi_1(\mathcal{A}^{+i}) > \Psi_1(\mathcal{A}^{+j})$.
Hence, (ii) holds.

\paragraph{Case 4: $|A_i|, |A_j| > 0$.}
In this case, the same argument as for $x < 1$ applies.

\vspace{3mm}
Therefore, (ii) is satisfied in all cases, and \Cref{thm:harmonic-computation} follows readily from Lemma~\ref{lem:yankee}.
\end{proof}

\section{More on Harmonic Welfare}
\label{app:harmonic-more}

In this section, we assume that all agents have additive valuations, and demonstrate some potential (and limits) of harmonic welfare.
Recall that $H_k = 1+\frac{1}{2}+\dots+\frac{1}{k}$ for each positive integer~$k$ and $H_0 = 0$.

\begin{definition}[MHW]
Given an instance with equal weights and additive valuations in which each agent's value for each good is an integer, an allocation $\mathcal{A}$ is a \emph{maximum harmonic welfare (MHW)} allocation if it maximizes the \emph{harmonic welfare} $\text{HW}(\mathcal{A}) := \sum_{i\in N} H_{v_i(A_i)}$.
\end{definition}

Note that in the unweighted setting, MHW allocations are the same as MWHW$_0$ allocations defined in \Cref{sec:harmonic}.
To establish the EF1 guarantee of MHW allocations, we will use the following lemma, which follows directly by inspecting the left and right Riemann sums of the function $f(x) = 1/x$.

\begin{lemma}
\label{lem:Riemann}
For integers $b \ge a \ge 1$, it holds that
$
\sum_{k = a}^b\frac{1}{k} > \ln\left(\frac{b+1}{a}\right)
$.
Moreover, if $a \ge 2$, then
$
\sum_{k = a}^b\frac{1}{k} < \ln\left(\frac{b}{a-1}\right)
$.
\end{lemma}

\begin{theorem}
\label{thm:MHW-EF1}
Suppose that all agents have equal weights and additive valuations, and each agent's value for each good is an integer.
Then, any MHW allocation satisfies EF1.
\end{theorem}

\begin{proof}
Let $\mathcal{A}$ be an MHW allocation, and assume for contradiction that for some pair of agents $i,j\in N$, it holds that $A_j\ne\emptyset$ and $v_i(A_i) < v_i(A_j) - v_i(g)$ for all $g\in A_j$.
Since each agent's value for each good is an integer, we have 
\begin{align}
\label{eq:EF1-integer}
v_i(A_i) \le v_i(A_j) - v_i(g) - 1
\end{align}
for all $g \in A_j$.
In particular, $v_i(A_j) \ge 1$, and so there exists $g\in A_j$ such that $v_i(g) \ge 1$.
Plugging such a good $g$ into \eqref{eq:EF1-integer}, we get $v_i(A_j) \ge 2$.
By \eqref{eq:EF1-integer} again, $v_i(A_j) \ge v_i(g) + 1$ for each $g\in A_j$.

Let $B = \{g \in A_j\colon v_i(g) > 0\}$; from the previous paragraph, we have $|B| \ge 2$.
Since $\mathcal{A}$ is an MHW allocation, moving any good $g$ from $A_j$ to $A_i$ cannot increase the harmonic welfare.
In particular, $v_j(g) > 0$ for all $g\in B$, which also means that $v_j(g) \le v_j(A_j) - 1$ for each $g\in B$.
We have
\begin{align*}
H_{v_i(A_i)} + H_{v_j(A_j)} \ge H_{v_i(A_i)+v_i(g)} + H_{v_j(A_j)-v_j(g)}
\end{align*}
for all $g\in B$.
Equivalently,
\begin{align*}
\frac{1}{v_j(A_j) - v_j(g) + 1} + \dots + \frac{1}{v_j(A_j)} 
\ge \frac{1}{v_i(A_i) + 1} + \dots + \frac{1}{v_i(A_i) + v_i(g)}.
\end{align*}
Combining this with \eqref{eq:EF1-integer} yields
\begin{align*}
\frac{1}{v_j(A_j) - v_j(g) + 1} + \dots + \frac{1}{v_j(A_j)} 
\ge \frac{1}{v_i(A_j) - v_i(g)} + \dots + \frac{1}{v_i(A_j) - 1}.
\end{align*}
Since $v_i(A_j) - v_i(g) \ge 1$ and $v_j(A_j) - v_j(g) + 1 \ge 2$ for every~$g\in B$, applying Lemma~\ref{lem:Riemann} to both sides, we get
\begin{align*}
\ln\left(\frac{v_j(A_j)}{v_j(A_j) - v_j(g)}\right) > \ln\left(\frac{v_i(A_j)}{v_i(A_j) - v_i(g)}\right)
\end{align*}
for all $g\in B$.
This implies that 
\begin{align*}
\frac{v_j(A_j)}{v_j(A_j) - v_j(g)} > \frac{v_i(A_j)}{v_i(A_j) - v_i(g)},
\end{align*}
which simplifies to $v_i(g)/v_j(g) < v_i(A_j)/v_j(A_j)$ for all $g\in B$.
As a result, we have
\begin{align*}
\frac{v_i(A_j)}{v_j(A_j)} 
= \frac{\sum_{g\in A_j}v_i(g)}{\sum_{g\in A_j}v_j(g)}
\le \frac{\sum_{g\in B}v_i(g)}{\sum_{g\in B}v_j(g)}
< \frac{v_i(A_j)}{v_j(A_j)};
\end{align*}
the first inequality holds because $\sum_{g\in A_j}v_i(g) = \sum_{g\in B}v_i(g)$ by definition of $B$ and $\sum_{g\in A_j}v_j(g) \ge \sum_{g\in B}v_j(g)$ due to the relation $B\subseteq A_j$.
This yields the desired contradiction.
\end{proof}

In spite of \Cref{thm:MHW-EF1}, a disadvantage of MHW compared to MNW is that MHW is not \emph{scale-invariant}---multiplying the value of a certain agent for every good by the same factor may change the MHW outcome.
Moreover, MHW is well-defined only when all utilities are integers.
Even though there is a natural extension of the harmonic numbers to the real domain given by $H_x = \int_0^1 \frac{1-t^x}{1-t} \diff t$ for $x\in\mathbb{R}$ \citep{Hintze19}, MHW defined via this extension is not guaranteed to satisfy EF1.

\begin{proposition}
\label{prop:MHW-extended}
There exists an instance with $n = 2$ agents with equal weights and additive valuations such that every MHW allocation, where MHW is defined based on the extended harmonic numbers $H_x = \int_0^1 \frac{1-t^x}{1-t} \diff t$, does not satisfy EF1.
\end{proposition}

\begin{proof}
Let $m = 3$, and assume that the agents' valuations are given by $v_1(g_1) = 0$, $v_1(g_2) = v_1(g_3) = 4$, $v_2(g_1) = 1.9$, and $v_2(g_2) = v_2(g_3) = 2$.
Clearly, in any MHW allocation, $g_1$ must be allocated to agent~$2$.
We consider the three possibilities.
\begin{itemize}
\item If agent~$1$ receives both $g_2$ and $g_3$, the harmonic welfare is $H_8 + H_{1.9} > 2.717 + 1.459 = 4.176$.
\item If agent~$1$ receives one of $g_2$ and $g_3$, the harmonic welfare is $H_4 + H_{3.9} < 2.084 + 2.061 = 4.145$.
\item If agent~$1$ receives neither $g_2$ nor $g_3$, the harmonic welfare is $H_0 + H_{5.9} < 0 + 2.435 = 2.435$.
\end{itemize}
Hence, the only MHW allocation gives $g_1$ to agent~$2$ and both $g_2$ and $g_3$ to agent~$1$.
However, agent~$2$ is not EF1 toward agent~$1$ with respect to this allocation.
\end{proof}

\end{document}